\documentclass[conference]{IEEEtran}

\usepackage[T1]{fontenc}
\usepackage{graphicx}
\usepackage{booktabs}
\usepackage{multirow}
\usepackage{amsmath,amssymb,amsthm}
\usepackage{enumitem}
\usepackage{hyperref}
\usepackage{xcolor}
\usepackage{colortbl}
\usepackage{tikz}
\usepackage{setspace}

\usepackage{stackengine}
\newcounter{sharedtabfn}

\usepackage{subcaption}
\usepackage{array}

\usepackage{fancyhdr}
\newenvironment{thisnote}
{\par\noindent\texttt\bgroup}
{\egroup\par}
\usepackage{cite}
\usepackage{comment}
\usetikzlibrary{arrows.meta,positioning,shapes,fit}

\usepackage{threeparttable}

\usepackage[norelsize,linesnumbered,ruled,lined,boxed,commentsnumbered,noend]{algorithm2e}
\SetKwInput{KwInput}{Input}
\SetKwInput{KwOutput}{Output}

\newtheorem{theorem}{Theorem}

\title{LLM-Based Multi-Agent Collaboration for Constrained Multi-Objective Container Placement}
\author{
\IEEEauthorblockN{
Suvarthi Sarkar\textsuperscript{1},
Abhishek Hazra\textsuperscript{2},
Akshat Mittal\textsuperscript{3},
Aryabartta Sahu\textsuperscript{1}
}
\IEEEauthorblockA{
\textsuperscript{1}Dept. of CSE, IIT Guwahati, India \quad
\textsuperscript{2}MIDC Lab, NIT Agartala, India \quad
\textsuperscript{3}Alumnus, IIT Guwahati, India
\\
}
\IEEEauthorblockA{
Email: s.sarkar@iitg.ac.in, ah.cse@faculty.nita.ac.in, m.akshat@alumni.iitg.ac.in, asahu@iitg.ac.in
}
}\date{}

\begin{document}

\maketitle
\thispagestyle{fancy} % CRITICAL: Forces the header onto page 1

% ============================================================
\begin{abstract}
Container placement in data centers must simultaneously minimize power consumption and maximize affinity preferences, while satisfying multi-resource capacity and anti-affinity constraints. Traditional approaches typically rely on fixed rules, which lack adaptability to dynamic cluster states and are difficult to extend for adaptive decision-making. On the other hand, meta-heuristic methods, although more flexible, are often computationally expensive, slower and prone to getting trapped in local optima. 
In this work, we propose an adaptive and efficient approach based on an LLM-driven multi-agent collaboration framework, where four specialized agents operate in a closed-loop ReAct cycle at each placement step. A Power Consumption Agent and an Affinity Agent debate over competing objectives, while a Placement Agent resolves conflicts using Tolerance-Gated Arbitration (TGA). A Rearrangement Agent further refines decisions through the Monotone Improvement Rule (MIR) in a post-placement refinement phase. All the agents reason over deterministic, feasibility-filtered candidate tables provided by the environment, ensuring that every proposed action inherently satisfies hard constraints. Based on experimental evaluation using the Google Cluster
Trace with a configuration of 100 applications and 25
machines, the proposed framework consumes 8.36\% less power and achieves 39.30\% higher affinity than the power-greedy baseline, 7.12\% less power and 10.85\% higher affinity than the best domain-specific heuristic (HOP-CAPM), and 4.26\% less power and 1.83\% higher affinity than the meta-heuristic approach (CPAAP), averaged over three independent runs. Ablation studies confirm that each architectural component, the multi-round debate mechanism, the TGA in the Placement Agent,
and the MIR in the Rearrangement Agent, contributes meaningfully to the overall performance, with removal of any single component degrading both power consumption and affinity satisfaction.
\end{abstract}

% ============================================================
\section{Introduction}

Modern data centers host large-scale containerized applications that are deployed and managed through orchestration platforms. In such systems, containers (referred to in this work as \emph{applications}) are placed onto physical machines, where they typically run for extended durations. As a result, placement decisions have long-term impact on energy consumption, resource utilization, and overall system efficiency \cite{sarkar,ilp}.

Operators increasingly aim to reduce power consumption as they consume more than 50\% of the data center costs \cite{cubic_power_model}. Along with power consumption reduction they also focus to placement preferences such as \emph{affinity} and \emph{anti-affinity}. Affinity represents performance or compliance benefits when certain applications are co-located or assigned to specific machines (e.g., due to hardware features, bandwidth, GPU availability, reliability history, or policy constraints). These are treated as soft constraints: violations are allowed but reduce performance or payoff. In contrast, anti-affinity represents strict incompatibility (e.g., missing hardware/software requirements, legal restrictions, security constraints) and must be treated as hard constraints. Applications cannot be placed on forbidden machines. Additionally, every placement must satisfy multi-resource feasibility constraints across CPU, memory, I/O, and network bandwidth \cite{sarkar,ilp}.

\begin{figure}[t]
    \centering
    \includegraphics[width=\linewidth]{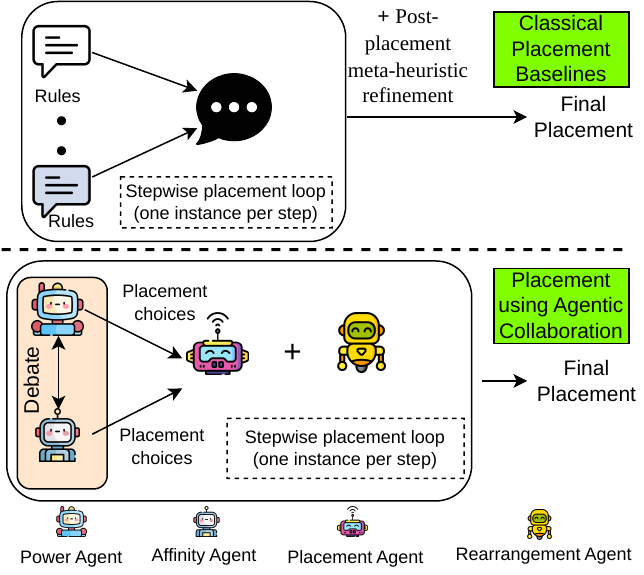}
\caption{\footnotesize \textbf{Comparison between classical placement pipelines and the proposed agentic collaboration framework. Classical approaches proceed in a stepwise fashion and may optionally include a global meta-heuristic refinement stage after the initial placement phase. In contrast, the proposed framework incorporates structured collaboration among multiple agents at each placement step. The \emph{Power Consumption Agent} and the \emph{Affinity Agent} debate and produce candidate placement options, where each option corresponds to assigning the current application to a specific machine. These candidate placements are then sent to the \emph{Placement Agent}, which chooses the most suitable option based on the aggregated evaluations of the agents while strictly satisfying all hard constraints. Following each placement, the \emph{Rearrangement Agent} determines whether local move or swap operations can further improve the objective without compromising feasibility. This collaborative process is repeated for every application instance until the final placement is obtained. Further details of the agents are given in \autoref{sec:agents}.}}
\label{fig:intro_agentic}
\vspace{-0.4cm}
\end{figure}

Power-aware placement under affinity and anti-affinity is fundamentally challenging. First, the decision space grows combinatorially with the number of applications and machines. Second, realistic power models are non-linear (e.g., cubic CPU utilization terms combined with linear resource components), making objective optimization non-trivial. Third, placement is inherently sequential: once an application is assigned to a machine, subsequent decisions must adapt to the existing configuration. Early suboptimal placements can propagate through the system, causing lock-in effects that degrade long-term power–affinity trade-offs.
In multi-objective settings such as this, there is typically no single optimal solution that simultaneously minimizes power and maximizes affinity. Instead, systems seek \emph{Pareto-optimal} solutions. A placement configuration is Pareto-optimal if no other feasible configuration exists that improves one objective (e.g., lower power) without worsening at least one other objective (e.g., affinity). Practical schedulers therefore aim to approach Pareto-efficient trade-offs while respecting strict feasibility constraints ~\cite{deb2001multi}.

Traditional container placement approaches rely primarily on hand-crafted heuristics. These methods typically consist of predefined decision rules, as shown in the top part of \autoref{fig:intro_agentic}. While such heuristics are computationally efficient and easy to implement, but usually do not perform best under all conditions \cite{swain}.
To address this limitation, meta-heuristic techniques have been explored. These approaches improve solution quality by exploring larger neighborhoods and performing iterative refinements after an initial placement \cite{sarkar}. But they are slow and might get stuck to local optima.
More recently, learning-based schedulers have been proposed to automatically learn placement policies from data. However, such methods typically require extensive training, careful tuning of model parameters, and may suffer from inference latency and scalability challenges in large and dynamic cluster environments \cite{dl_prob, rl_prob, single_agent}.

Large Language Models (LLMs) have recently demonstrated strong sequential reasoning capabilities, making them well suited to adaptive scheduling. Unlike fixed-rule heuristics, an LLM-based scheduler can reason jointly over the current cluster state, prior placement history, and multiple competing objectives. Prior work has studied single-LLM schedulers for placement \cite{bai25}, where the model acts as a centralized decision-maker. However, in such settings, a single reasoning process can not simultaneously optimize for both power minimization and affinity maximization, without any explicit architectural mechanism to ensure balanced treatment of these competing goals. In this paper, we investigate whether a multi-agent design can overcome this limitation. We decompose the application placement into specialized objective-specific agents that argue independently, while a neutral arbitrator reconciles their recommendations, as shown in the bottom part of \autoref{fig:intro_agentic}. We transfer this principle from the evaluation domain to constrained multi-objective optimization, recasting container placement as a structured collaboration among agents with explicitly separated roles. Concretely, we propose a four-agent closed-loop framework in which all agents reason over a shared, deterministic candidate table generated by the environment at each placement step. 

The following are the contributions of the work:

\begin{itemize}

  \item We propose a novel multi-agent collaborative approach for bi-objective container placement in which a Power Consumption Agent and an Affinity Agent conduct structured debate, a Placement Agent arbitrates through Tolerance-Gated Arbitration (TGA), and a Rearrangement Agent refines placements through the Monotone Improvement Rule (MIR). All the agents are constrained to deterministic, feasibility-filtered candidate tables, guaranteeing hard-constraint compliance by construction. 
  \item We analyze the proposed framework theoretically. We bounds the power overhead of TGA relative to the power greedy baseline in terms of the tolerance parameter and the MIR recovery mechanism (in Theorem~1). We proves that the affinity achieved by TGA with MIR is at least as high as that of any power-greedy heuristic (in Theorem~2). Together, these results identify the conditions under which the proposed framework Pareto-dominates greedy heuristics on both power consumption and affinity.
    %\item We provide formal guarantees for the approach. We bounds the power overhead of TGA relative to the greedy baseline as a function of the tolerance parameter and MIR recovery (in Theorem 1). We proves that affinity under TGA and MIR is at least as high as under any power-greedy heuristic ( in Theorem 2). Together, these establish conditions under which the framework mostly Pareto-dominates greedy on both objectives. 
   \item We evaluate our propose multi-agent collaborative approach on the Google Cluster trace against four baseline categories: general heuristics (\textit{First Fit}, \textit{Best Fit}, \textit{P-Greedy}), domain-specific heuristics (\textit{HOP-CAPM} ~\cite{swain}), meta-heuristics (\textit{CPAAP} ~\cite{sarkar}), and alternative LLM architectures (single-agent, two-agent, three-agent, confidence scoring). The proposed approach achieves the best power--affinity trade-off across all the approaches. We further provide ablation studies over debate rounds, candidate list size, tolerance, and LLM model choice.
\end{itemize}

\section{Literature Review}

Application-to-machine placement have been extensively studied in cloud and cluster management. Many classical approaches rely on heuristic rules, which are computationally efficient but often rigid under 
dynamic system conditions.

Swain et al.~\cite{swain} proposed \textit{HOP-CAPM}, a heuristic scheduling framework that groups machines and applications into multiple subgroups based on hard constraints. Tasks are prioritized using a predefined set of rules, and each task is assigned to a compatible machine according to those rules. While the method is computationally fast, its reliance on fixed priority rules limits adaptability to changing workload characteristics and dynamic cluster states.
Sarkar et al.~\cite{sarkar} introduced \textit{CPAAP}, which augments a heuristic seed with a genetic algorithm (GA) search to improve placement quality. The approach explores better Pareto trade-offs between competing objectives. 
Although the GA stage can improve solution quality compared to pure heuristics, it introduces significantly higher 
runtime overhead and does not provide convergence guarantees.
Long et al.~\cite{ilp} formulate container placement as an integer linear programming (ILP) problem, referred to as 
GCCS. ILP-based approaches provide optimality guarantees under the assumed model. However, they typically rely on simplified linear power models and scale poorly when the placement problem incorporates nonlinear power functions or large cluster sizes.

Recently, LLMs have been explored for resource management and scheduling tasks. Bai et al.~\cite{bai25} used an LLM to assist Kubernetes scheduling, focusing on single-objective latency optimization. Jadhav et al.~\cite{single_agent} evaluateed LLM-based reasoning for multi-objective HPC scheduling, but treated the model as a single monolithic scorer without structural separation of competing objectives. ChatHPC~\cite{yin2025chathpc} and LLMSched~\cite{zhu2025llmsched} investigated LLM-driven resource management but neither considered multi-agent collaboration nor provides solution quality guarantees. In the broader multi-agent LLM literature, debate-based frameworks have been shown to improve reasoning quality over single-agent designs~\cite{du2024debate, chateval}, role-specialized collaboration has been applied to structured task decomposition~\cite{autogen, chatdev}, and verbal scratchpad memory has been shown to improve iterative decision quality~\cite{reflexion}. 

However, none of these frameworks addressed constrained combinatorial optimization under hard feasibility requirements or provides formal per-step guarantees on objective quality. Our work bridges this gap by adapting structured multi-agent debate from the evaluation domain~\cite{chateval} to placement optimization domain. To the best of our knowledge, this is the first work combining LLM-based multi-agent debate with formal guarantees for multi-objective container placement under non-linear power models and hard anti-affinity constraints.

\begin{figure*}[t]
\centering
\includegraphics[scale=0.8]{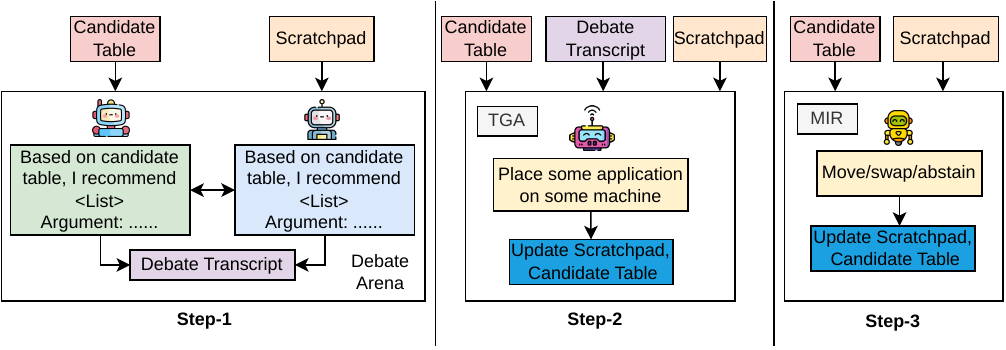}
\vspace{-0.4cm}
\caption{\footnotesize \textbf{Multi-Agent Debate Scheduler architecture showing the
  three-step closed-loop ReAct procedure.
  \textit{Step~1 (Debate):} The Power Consumption Agent and Affinity Agent receive the candidate
  table and scratchpad. They debate for $N$
  configurable rounds, each outputting $K$ ranked recommendations with written arguments.
  \textit{Step~2 (Placement):} The Placement Judge reads the full
  transcript, candidate table, and scratchpad, then select exactly one action. \textit{Step~3 (Rearrangement):} The Rearrangement Agent might swap or move containers for better solution.}}
  \vspace{-0.4cm}
\label{fig:architecture}
\end{figure*}
\nocite{asahu_1,asahu_2}

\section{Problem Formulation}

%\subsection{Problem Formulation}
%\label{sec:formulation}

We model container placement as a multi-stage assignment
problem.
Let $t$ denote the placement step index.  At step $t$, the system
is in state $\mathcal{S}_t$ consisting of:

\begin{itemize}
\item A set of machines
      $\mathcal{M} = \{1, 2, \dots, M\}$.  Each machine $m$ has
      capacity $C_m^x$ for resource dimension
      $x \in \mathcal{R} = \{\textit{cpu}, \textit{mem}, \textit{io},
      \textit{bw}\}$.
\item A set of unplaced container instances set $\mathcal{I} = \{1, 2, \dots, I\}$. Each instance $i$ has demand vector
      $\mathbf{r}_i = (r_i^x)_{x \in \mathcal{R}}$.
\item An affinity matrix $G$ and an anti-affinity matrix $F$,
      both defined over instance--machine pairs.
\end{itemize}

A placement $\pi : \mathcal{I} \to \mathcal{M}$ assigns every
instance to a machine.  We now define the constraints and
objectives that governs $\pi$.

\subsection{Anti-affinity (hard constraint)}
We define a binary matrix $F(i,m) \in \{0,1\}$.  If
$F(i,m) = 1$, instance $i$ is forbidden from machine $m$.  Every
valid placement must satisfy:
\begin{equation}
F\bigl(i,\,\pi(i)\bigr) = 0,
\quad \forall\, i \in \mathcal{I}.
\label{eq:antiaffinity}
\end{equation}

\subsection{Affinity (soft preference)}
We define a binary payoff matrix $G(i,m) \in \{0,1\}$.  If
$G(i,m) = 1$, placing instance $i$ on machine $m$ is preferred.
Violations are allowed but reduce total payoff.  The total
affinity payoff under $\pi$ is:
\begin{equation}
A(\pi) = \sum_{i \in \mathcal{I}} G\bigl(i,\,\pi(i)\bigr).
\label{eq:affinity}
\end{equation}

\subsection{Capacity (hard constraint)}
Each machine $m$ must have enough capacity in every resource
dimension to serve its assigned instances:
\begin{equation}
\sum_{i:\,\pi(i)=m} r_i^x \;\leq\; C_m^x,
\quad \forall\, m \in \mathcal{M},\;
      \forall\, x \in \mathcal{R}.
\label{eq:capacity}
\end{equation}

\subsection{Power Consumption model}
Each machine $m$ has idle power consumption $P_m^{\text{idle}}$ and peak power $P_m^{\text{max}}$. For resource dimension $x \in \mathcal{R}$,
let $L_m^x = \sum_{i:\,\pi(i)=m} r_i^x$ be the total load and
$u_m^x = L_m^x / C_m^x$ the utilization ratio. We collect the
four utilization ratios into a vector
$\mathbf{u}_m = (u_m^x)_{x \in \mathcal{R}}$.

Machine power consumption under placement $\pi$ is
\begin{equation}
P_m(\pi) = P_m^{\text{idle}}
  + \phi(\mathbf{u}_m)
    \bigl(P_m^{\text{max}} - P_m^{\text{idle}}\bigr),
\label{eq:machinepower}
\end{equation}
\nocite{roy_hipc}
where the load factor $\phi$ combines a cubic CPU term with linear
contributions from the remaining resources:
\begin{equation}
\phi(\mathbf{u}_m) = \min\!\Bigl(1,\;\;
  \alpha_1 \bigl(u_m^{\text{cpu}}\bigr)^{\!3}
  + \alpha_2\, u_m^{\text{io}}
  + \alpha_3\, u_m^{\text{bw}}
  + \alpha_4\, u_m^{\text{mem}}
\Bigr).
\label{eq:phi}
\end{equation}
The cubic CPU term dominates: doubling CPU utilization increases
the CPU power consumption component eightfold.  We set
$\alpha_1 = 0.70$, $\alpha_2 = 0.10$, $\alpha_3 = 0.05$, and
$\alpha_4 = 0.15$ following~\cite{cubic_power_model, tarafdar}.
Total cluster power is
\begin{equation}
P(\pi) = \sum_{m \in \mathcal{M}} P_m(\pi).
\label{eq:totalpower}
\end{equation}

\subsection{Objectives}
A placement $\pi$ is \emph{feasible} if it satisfies
~\autoref{eq:antiaffinity}, \ref{eq:capacity}. Let $\Pi$ denote the set of all feasible placements. We solve the
following prioritized bi-objective problem: minimize $P(\pi)$ as the primary objective and maximize $A(\pi)$ as the secondary objective, over $\pi \in \Pi$.

\section{LLM-Based ReAct Multi-Agent Placement Framework} \label{sec:framework}

We design an LLM-driven placement framework for assigning
container instances to machines in a cluster.  The agentic framework operates in a closed loop with a deterministic placement environment.  At each step, it observes the current cluster state, reasons about
power and affinity trade-offs, and outputs a structured placement
action.  The environment enforces all hard constraints and
computes updated metrics after each action.  This separation
ensures correctness while allowing the LLM to focus on
decision-making.

\nocite{asahu_3,asahu_4}
\autoref{fig:architecture} shows the overall architecture of the proposed agents, and Algorithm \ref{alg:placement} summarizes the full workflow. The framework operates as a ReAct-style loop~\cite{react}: at each step, the agents first \textit{reason} over the shared cluster state, comprising the current placement, machine utilizations, and scratchpad history, and then \textit{act} by producing a structured placement action. This separation of reasoning and acting allows each LLM invocation to focus on trade-off evaluation, while the environment strictly enforces all hard constraints and recomputes the updated metrics after every action, thus guaranteeing correctness by construction. The procedure terminates when all instances are placed, no feasible candidate exists, or the judge returns $\texttt{Stop}$.

\begin{algorithm}[t]
\caption{Closed-Loop ReAct Multi-Agent Placement}
\label{alg:placement}

\KwInput{Machines $\mathcal{M}$, instances $\mathcal{I}$, affinity $G$, anti-affinity $F$, $N$: debate rounds , $K$: choices}
\KwOutput{Feasible placement $\pi$}

$t \leftarrow 0$\;
$\pi_t \leftarrow \emptyset$\;

\tcp{Phase 1: Build tables}
$\texttt{C}_0 \leftarrow \text{initialize \textsc{CandidateTable} to NULL}$\;
$\texttt{S}_0 \leftarrow \text{initialize \textsc{Scratchpad} to NULL}$\;

\While{$\exists\, i \in \mathcal{I}$ \textnormal{with instances remaining}}{
    
    \tcp{Phase 2: Debate}
    $\texttt{T}_t \leftarrow \text{\textsc{Transcript} at step } t$\;
    
    \For{$r \leftarrow 1$ \KwTo $N$}{
        $\mathbf{p}_K \leftarrow \textsc{PowerAgent}(\texttt{C}_t,\texttt{S}_t,\texttt{T}_t)$\;
        $\texttt{T}_t \leftarrow \texttt{T}_t \oplus \mathbf{p}_K$ \tcp{\scriptsize{update $\texttt{T}_t$ with its preferred candidate and argument}}
        
        $\mathbf{a}_K \leftarrow \textsc{AffinityAgent}(\texttt{C}_t,\texttt{S}_t,\texttt{T}_t)$\;
        $\texttt{T}_t \leftarrow \texttt{T}_t \oplus \mathbf{a}_K$ \tcp{\scriptsize{update $\texttt{T}_t$ with its preferred candidate and argument}}
    }
    
    \tcp{Phase 3: Placement Agent}
    Select $\texttt{Place}$, $\texttt{Delay}$, or $\texttt{Stop}$ based on $(\texttt{C}_t,\texttt{S}_t,\texttt{T}_t,K)$\;
    Update $\texttt{C}_t$ and $\texttt{S}_t$\;
    
    \tcp{Phase 4: Rearrangement}
    Select $\texttt{Move}$, $\texttt{Swap}$, or $\texttt{Abstain}$ based on $(\texttt{C}_t,\texttt{S}_t)$\;
    Update $\texttt{S}_t$ and $\texttt{C}_t$ if $\texttt{Move}$ or $\texttt{Swap}$\;
    
    $t \leftarrow t+1$\;
}

\Return{$\pi$}
\end{algorithm}

\subsection{Environment-constructed Data Structures}
\label{sec:tables}

At each placement step $t$, the environment holds the current placement $\pi_t$, machine utilizations $\mathbf{u}_m$ for all $m \in \mathcal{M}$, and remaining instance counts. From this
state, it constructs three data structures. 

\subsubsection{Candidate Table}
At each placement step $t$, the environment builds a table of all feasible placement options. It pairs every eligible application with every machine. For each pair $(i, m)$, the environment applies two filters. First, it rejects the pair if $F(i, m) = 1$ (anti-affinity). Second, it rejects the pair if machine $m$ lacks capacity in any resource dimension (\autoref{eq:capacity}). Pairs that pass both
filters are feasible.

For each feasible pair, the environment computes one value:
the incremental power consumption for placing container $i^{th}$ on $m^{th}$ machine ($\Delta P_{i,m}$) from \autoref{eq:machinepower}, \ref{eq:phi}. It then sorts all feasible rows in ascending order of $\Delta P$, since minimizing power is the primary objective of the problem, and breaks ties in descending order of $G$ (as shown in \texttt{line 4} of Algorithm \ref{alg:placement}). The table is updated after every placement (\texttt{line 16, 19}). Suppose instance $i$ is placed on machine $m$ at step $t$. Two things change. First, $\mathbf{u}_{m}$ increases because $m$ now carries more load. This changes $\Delta P_{i, m}$ for every remaining application on that machine. Second, the placed instance $i$ is removed for all the machines from the table.

The updated table has one row fewer per machine. The environment recomputes $\Delta P$ values for that machine where container is placed, and re-sorts. Over the full scheduling run, the candidate table shrinks as containers get placed and machines fill up. Agents select a feasible placement from this table based on its logic. So, they cannot propose an infeasible placement.

\subsubsection{Scratchpad}
The scratchpad is a compact text block that provides trajectory context. It contains: current step $t$, total placements, total swaps, and the last few placement records (each showing instance, machine, $G$, and $\Delta P$).

\subsection{Agent Roles}
\label{sec:agents}

\nocite{asahu_5}
We use four LLM-prompted agents. All operate over the same shared state (as shown in \autoref{fig:architecture}). Each agent outputs a single structured action per invocation. Each agent receives a fixed instruction prompt specific to its role and a set of dynamic inputs from the environment. It makes its decision based on these inputs. 

\subsubsection{Power Consumption Agent} It reads the candidate table, scratchpad and debate argument from Affinity Agent (from second round onwards). It uses the candidate table which is ranked based on $\Delta P$ ascending. It outputs the top-$K$ recommendations and a written argument explaining its reasoning (\texttt{line 10-11}). $\mathbf{p}_K$ represents the preference placement list and corresponding arguments of the Power Agent and $\oplus$ represents the concatenation operation on the transcript.\textsuperscript{\ref{note_2}}

\subsubsection{Affinity Agent}
It reads the candidate table, scratchpad and debate argument from Power Consumption Agent (from second round onwards). It scans the candidate table for applications with $G(i,m) = 1$ and promotes them to the top of its ranking. It breaks ties by lowest $\Delta P$. It also outputs $K$ recommendations and an argument emphasizing co-location benefits (\texttt{line 12-13}). If no $G = 1$ candidate exists, it ranks the candidate by $\Delta P$ alone. Across debate rounds, both the Power Consumption Agent and the Affinity Agent adapt their preference lists in response to each other’s feedback, thereby progressively refining their candidate rankings. $\mathbf{a}_K$ represents the preference placement list and corresponding arguments of the Affinity Agent and $\oplus$ represents the concatenation operation on the transcript.\footnote{\label{note_2} The example prompt is given in our \href{https://anonymous.4open.science/r/AgenticAI-in-Constrained-Multi-Objective-Container-Placement-E685/README.md}{Github Repository}.}

\begin{comment}
    
\begin{table}[t]
\centering
\begin{tabular}{|p{0.45\textwidth}|}
\hline
\cellcolor{green!15}
\begin{thisnote}
You are an affinity-maximizing cluster scheduler. Your ONLY goal is maximum affinity.

At this step, you must choose exactly ONE action:

- Place ONE instance of an app onto ONE machine, OR

- Delay (only if no feasible placement exists), OR

- Stop (only if all instances are scheduled).

\textbf{HARD CONSTRAINTS} (strict, must satisfy now):
Same as \autoref{tab:power_prompt}

\textbf{GOAL:} Maximize affinity (prefer affinity=1 placements). Power is IGNORED.

\textbf{DECISION POLICY} — Affinity-only (strict):
- All rows in the candidate table are already feasible.

- First choose candidates with affinity=1. Among those, pick the one with smallest delta\_total as tiebreaker.

- If no affinity=1 candidate exists, pick any candidate (smallest delta\_total).

- Ignore power minimization completely.

\textbf{Scratchpad:}
\{scratchpad\}

\textbf{Candidate Table:}
\{Candidate Table\}
\end{thisnote}
\\
\hline
\end{tabular}
\caption{Prompt of Affinity Agent}
\label{tab:affinity_prompt}
\end{table}
\end{comment}

\subsubsection{Placement Agent} 

It receives the debate transcript, candidate table, and scratchpad, and applies Tolerance-Gated Arbitration (TGA) to select a placement pair $(i,m)$ (\texttt{line 14-16}). Instead of always choosing the candidate with minimum power increment, it first considers all feasible candidates whose power increment lies within a tolerance gate of the minimum feasible power. Within this gated set, candidates with $G(i,m)=1$ are prioritized, while the final decision is also guided by the recommendations recorded in the debate transcript. The selection rule is given by

{\footnotesize{
\begin{equation}
(i^\star,m^\star)=
\begin{cases}
\arg\min\limits_{\substack{(i,m)\in\mathcal{F}\\
G(i,m)=1,\;
\Delta P_{i,m}\le (1+\varepsilon)\Delta P_{\min}}}
\Delta P_{i,m},
\\
\hfill \text{if such a feasible }(i,m)\text{ exists,}
\\[4pt]
\arg\min\limits_{\substack{(i,m)\in\mathcal{F}\\
\Delta P_{i,m}\le (1+\varepsilon)\Delta P_{\min}}}
\Delta P_{i,m},
\\
\hfill \text{otherwise,}
\end{cases}
\label{eq:tga}
\end{equation}}}

where $\mathcal{F}$ denotes the set of feasible candidate placement pairs $(i,m)$ satisfying all hard constraints, and $\Delta P_{\min}=\min_{(i,m)\in\mathcal{F}} \Delta P_{i,m}$, considering $(i^\star,m^\star)$ is the selected pair. When multiple comparable candidates remain, the final decision is further guided by the debate transcript. When $\varepsilon=0$, the rule reduces to pure power-greedy selection. As $\varepsilon$ increases, more affinity-preferred candidates fall within the gate.\textsuperscript{\ref{note_2}}

\begin{comment}

\begin{table}[t]
\centering
\begin{tabular}{|p{0.45\textwidth}|}
\hline
\cellcolor{red!15} 
\begin{thisnote}    
You are an expert power-aware cluster scheduler. At this step, you must choose exactly ONE action:

- Place ONE instance of an application onto ONE machine, OR

- Delay (only if no feasible placement exists), OR

- Stop (only if all instances are scheduled).

\textbf{HARD CONSTRAINTS} (strict, must satisfy now):
Same as \autoref{tab:power_prompt}

\textbf{GOALS} (consider BOTH; do NOT combine into a single scalar score):

- Minimize incremental cluster power delta\_total (lower is better) (main goal)

- Maximize affinity payoff A in 0 or 1 (1 is better) (soft goal)

\textbf{POWER MODEL}: Same as \autoref{tab:power_prompt}

\textbf{DECISION POLICY:}

- All rows in the candidate table are already feasible.

- Let DPmin = smallest delta\_total among all candidates.

- epsilon (Tolerance)

- If any candidate has affinity=1 AND delta\_total <= DPmin + epsilon:
    choose the affinity=1 candidate with smallest delta\_total
  else:
    choose the candidate with the smallest delta\_total

\textbf{Scratchpad:}
\{scratchpad\}

\textbf{Candidate Table:}
\{Candidate Table\}

\end{thisnote}
\\
\hline
\end{tabular}
\caption{Prompt of Placement Agent}
\label{tab:placer_prompt}
\end{table}
\end{comment}

\subsubsection{Rearrangement Agent}

The Rearrangement Agent performs local refinement of the placement after each action selected by the Placement Agent. After every placement step, the system updates both the scratchpad and the candidate table to reflect the current cluster state. The agent selects at most one operation (\texttt{Move} or \texttt{Swap} or \texttt{Abstain}) according to the Monotone Improvement Rule (MIR) (\texttt{line 17-19}).
The MIR ensures that a rearrangement is accepted only if it improves the placement quality without worsening any objective. Specifically, the selected operation must either reduce power consumption without decreasing affinity, or increase affinity without increasing power. Formally,

\begin{equation}
(\Delta P(o) < 0 \;\wedge\; \Delta A(o) \geq 0) \;\vee\; (\Delta A(o) > 0  \;\wedge\; \Delta P(o) \leq 0 ),
\label{eq:mir}
\end{equation}

where $o$ is the operations taken by the rearrangement agent. Consequently, the rearrangement stage produces a monotonically improving refinement trajectory, gradually improving the placement quality while preserving feasibility.\textsuperscript{\ref{note_2}}

When rearrangement is triggered, the scratchpad and the candidate table can be updated twice in a single step: (a) before placement by the Placement Agent, so that the updated state is available to the debating agents and the Placement Agent (\texttt{line 16}); and (b) after placement, so that it can be used by the Rearrangement Agent (\texttt{line 19}).
\begin{comment}
\begin{table}[t]
\centering
\begin{tabular}{|p{0.45\textwidth}|}
\hline
\cellcolor{blue!15}
\begin{thisnote}
You are an expert post-placement optimizer (Rearrangement Agent).

After the last placement, you may improve the current schedule by ONE operation.

Allowed operations:

1) Move(instance\_id, app\_id, 
from\_machine, to\_machine)

2) Swap(instance\_id1, app\_id1, machine1, instance\_id2, app\_id2, machine2)

3) NoChange

\textbf{HARD CONSTRAINTS} (strict, must satisfy now):
Same as \autoref{tab:power_prompt}

\textbf{POWER MODEL}: Same as \autoref{tab:power_prompt}

\textbf{IMPROVEMENT RULE:}
Accept ONLY if:
  1) dP\_cluster < 0 AND dAff >= 0   (power drops, affinity doesn't worsen), OR
  
  2) dP\_cluster <= 0 AND dAff > 0   (power same or better, affinity improves)

Selection: most negative dP\_cluster first; tie-break by highest dAff; then first row.

ABSOLUTE RULE: If every row has dP\_cluster > 0, output NoChange.

\textbf{Scratchpad:}
\{scratchpad\}

\textbf{Candidate Table:}
\{Candidate Table\}

\end{thisnote}
\\
\hline
\end{tabular}
\caption{Prompt of Rearrangement Agent}
\label{tab:rearranger_prompt}
\end{table}
\end{comment}

\subsection{Debate Protocol}
\label{sec:debate}

The Power Consumption and Affinity agents debate for $N$ configurable rounds. Each round produces two blocks: one from each agent. The transcript accumulates $2N$ blocks total. Each advocate outputs $K$ ranked recommendations per round, along with the argument behind it. The agent counter the opposing agent's argument from the previous round.

\section{Formal Analysis}
\label{sec:formal}

%--------------------------------------------------------------------
\subsection{Heuristics as Fixed-Rule Policies}
%--------------------------------------------------------------------

A rule-based heuristic defines a fixed mapping
$\sigma_H : \pi_t \to (i, m)$ that selects one
candidate from the table at every step $t$.  The mapping is
determined before scheduling begins and does not change during execution. These policies share two structural limitations:

\begin{itemize}
    
\item \textit{Fixed rule:} The weights $(w_1, w_2)$ or the ordering logic is set once. It does not adapt to the changing cluster state. At step $t = 0$ (empty cluster) and step $t$ near $|\mathcal{I}|(=I)$ (nearly full cluster), the heuristic applies the same rule. Even for adaptive heuristics, designing suitable weight adjustments in response to changing cluster conditions can be highly challenging.

\item \textit{No revision:} Once instance $i$ is placed on machine $m$, the assignment $\pi(i) = m$ is permanent. The heuristic do not move or swap instances to correct earlier decisions, as the complexity is huge. Because there might be scenarios where ordering the application selection outputs a better overall placement.
\end{itemize}

%--------------------------------------------------------------------
\subsection{Why the Environment Demands Adaptive Rules?}
%--------------------------------------------------------------------

The candidate table $\texttt{C}_t$ changes at every step.
After placing instance $i$ on machine $m^*$ at step $t$, the utilization $\mathbf{u}_{m^*}$ increases.  Under the cubic power model, the marginal cost of the next placement on $m^*$
is:
\begin{equation}
\frac{\partial \Delta P}{\partial u^{\text{cpu}}}
= 3\alpha_1 (u_{m^*}^{\text{cpu}})^2
  \times (P_{m^*}^{\text{max}} - P_{m^*}^{\text{idle}}).
\label{eq:marginal}
\end{equation}
This grows quadratically with utilization.  At $u = 0.2$, the
sensitivity is $3 \times 0.70 \times 0.04 = 0.084$.  At
$u = 0.8$, it is $3 \times 0.70 \times 0.64 = 1.344$, a
$16\times$ increase.

A fixed rule that was good at $u = 0.2$ may be bad at
$u = 0.8$.  For example, a greedy heuristic based on power consumption (\textit{P-Greedy}) spreads load evenly. Early on, this is
cheap. But it pushes all machines toward moderate utilization where the cubic curve steepens. Later placements face uniformly high $\Delta P$ values with no cheap option left.

An adaptive policy would consolidate early (pack some machines,
keep others empty) to preserve low-utilization machines for
later.  A fixed heuristic cannot make this decision because it
does not reason about trajectory.

%--------------------------------------------------------------------
\subsection{TGA as an Adaptive Policy}
%--------------------------------------------------------------------

TGA (\autoref{eq:tga}) adapts to the environment at every
step through three mechanisms:

\begin{itemize}
     
\item \textit{Dynamic Gate:} The tolerance limit scales with the current best candidate. When machines are lightly loaded, the tolerance limit is small. When machines are heavily loaded, the tolerance limit is large, so the gate widens. The rule automatically adjusts to cluster state.

\item \textit{Debate Transcript:} The Power Consumption and Affinity agents provide arguments based on the current candidate table $\texttt{C}_t$ (at step $t$) and scratchpad $\texttt{S}_t$ (at step $t$). Their reasoning changes at every step because the inputs change. %A heuristic has no equivalent input.

\item \textit{Scratchpad Memory:} The Placement Agent reads the last few placement records. It can detect patterns: ``three consecutive placements on $m_3$, utilization rising fast.''  A heuristic sees only the current table, not the history.
\end{itemize}

\begin{comment}

Formally, the information available to each policy is:
%
\begin{align}
{Inf}_t^{\text{heur}} &= \mathcal{C}_t,
\label{eq:infoheur}\\
{Inf}_t^{\text{TGA}}  &= \mathcal{C}_t \cup
  \mathcal{S}_t \cup \mathcal{T}_t,
\label{eq:infotga}
\end{align}
%
where $\mathcal{S}_t$ is the scratchpad and $\mathcal{T}_t$ is
the debate transcript.  Since
$Inf_t^{\text{heur}} \subset Inf_t^{\text{TGA}}$,
TGA operates on a strict superset of information.
\end{comment}

\subsection{Formal Guarantees of Agentic Collaboration Architecture}

We now prove two results. The first bounds the power cost of
preferring affinity. The second shows affinity never falls
below the greedy baseline.

\paragraph{Notation}
Let $S$ denote the total placement steps.  At step $t$, let
$\Delta P_{\min}^t = \min_{(i,m) \in \mathcal{C}_t}
\Delta P_{i,m}$ be the cheapest candidate in the table.  Let
$|\Delta P(\rho_t)|$ be the power saved by the Rearrangement
Agent at step $t$ (zero if MIR is skipped). Let $P_{\text{greedy}}(S)$ and
$A_{\text{greedy}}(S)$ denote the total power and affinity
after $S$ steps under the power-greedy heuristic (\textit{P-Greedy}), which
selects $\arg\min \Delta P_{i,m}$ at every step with no rearrangement. We compare with \textit{P-Greedy} as minimizing power consumption is the primary objective.

\begin{theorem}[Power consumption overhead bound]
\label{thm:overhead}
The total power under TGA + MIR satisfies:
\begin{equation}
P_{\textup{ours}}(S) \leq P_{\textup{greedy}}(S)
  + \varepsilon \sum_{t=1}^{S} \Delta P_{\min}^t
  - \sum_{t=1}^{S} |\Delta P(\rho_t)|.
\label{eq:overhead}
\end{equation}
\end{theorem}

\begin{proof}
Let's consider step $t$. The \textit{P-Greedy} selects the candidate with $\Delta P_{\min}^t$. TGA selects a candidate with $\Delta P_{i,m} \leq (1 + \varepsilon)\,\Delta P_{\min}^t$ (\autoref{eq:tga}). The difference between TGA's choice and
\textit{P-Greedy}'s choice is at most:
\begin{equation}
(1 + \varepsilon)\,\Delta P_{\min}^t - \Delta P_{\min}^t
= \varepsilon \cdot \Delta P_{\min}^t.
\label{eq:perstep}
\end{equation}
After placement, MIR applies at most one operation satisfying \autoref{eq:mir}. Any accepted operation has $\Delta P(\rho_t) \leq 0$, reducing total power by $|\Delta P(\rho_t)|$. Summing \autoref{eq:perstep} over all $T$ steps and
subtracting the MIR recoveries:
\begin{equation}
P_{\text{ours}}(S) - P_{\text{greedy}}(S)
\leq \varepsilon \sum_{t=1}^{S} \Delta P_{\min}^t
   - \sum_{t=1}^{S} |\Delta P(\rho_t)|
\end{equation}
Rearranging gives \autoref{eq:overhead}.
\end{proof}

The bound consists of two terms. The first,
$\varepsilon \sum_{t=1}^S \Delta P_{\min}^t$, is the maximum power
that TGA can spend on affinity preferences across the entire
run.  The second, $\sum_{t=1}^S |\Delta P(\rho_t)|$, is the power
that MIR recovers through rearrangement.
%If MIR recovery exceeds TGAspending, i.e. $\varepsilon \sum_{t=1}^S \Delta P_{\min}^t < \sum_{t=1}^S |\Delta P(\rho_t)|$then $P_{\text{ours}}(S) < P_{\text{greedy}}(S)$, then our framework uses less power than greedy.

\begin{theorem}[Affinity dominance]
\label{thm:affinity}
The total affinity under TGA + MIR is at least as high as
under the power-greedy heuristic:
\begin{equation}
A_{\textup{ours}}(S) \geq A_{\textup{greedy}}(S).
\label{eq:affdominance}
\end{equation}
\end{theorem}

\begin{proof}
Let's consider step $t$.  The greedy heuristic selects
$\arg\min \Delta P_{i,m}$.  Let $G_{\text{greedy}}^t$ denote
the affinity of greedy's choice, and $G_{\text{TGA}}^t$ the
affinity of TGA's choice.

If no candidate with $G(i,m) = 1$ lies within the gate, TGA
selects $\arg\min \Delta P_{i,m}$, identical to greedy.
Therefore $G_{\text{TGA}}^t = G_{\text{greedy}}^t$.

If a candidate with $G(i,m) = 1$ lies within the gate, TGA
selects it with $G_{\text{TGA}}^t = 1$.  Greedy may select a
candidate with $G_{\text{greedy}}^t = 0$ or $1$.  In either
case, $G_{\text{TGA}}^t \geq G_{\text{greedy}}^t$.

At no step does TGA select a candidate with lower affinity than
power greedy heuristics. After the placement phase, the total affinity from
TGA alone satisfies:
\begin{equation}
\sum_{t=1}^{S} G_{\text{TGA}}^t
\geq \sum_{t=1}^{S} G_{\text{greedy}}^t
= A_{\text{greedy}}(S).
\label{eq:tgaaff}
\end{equation}
During rearrangement, MIR accepts only operations with
$\Delta A(\rho_t) \geq 0$ (\autoref{eq:mir}).  Affinity cannot decrease. It can only increase or stay the same.
Therefore, the final affinity is:
\begin{equation}
A_{\text{ours}}(S)
= \sum_{t=1}^{S} G_{\text{TGA}}^t
  + \sum_{t:\,\Delta A(\rho_t) > 0} \Delta A(\rho_t)
\geq A_{\text{greedy}}(S).
\label{eq:finalaffinity}
\end{equation}
\end{proof}

The proof follows from two facts.  TGA never picks lower
affinity than greedy at any step.  MIR never reduces affinity
during rearrangement.  Together, affinity can only match or
exceed the greedy baseline.

Combining both theorems: if MIR recovers at least what TGA
spends ($\sum_t |\Delta P(\rho_t)| \geq \varepsilon \sum_t
\Delta P_{\min}^t$), then our proposed approach achieves both lower power and higher affinity than greedy. This is a strict Pareto
improvement.

\subsection{Why Greedy Heuristic can not match the Agentic Collaboration Architecture?}

The greedy heuristic is structurally can not replicate
this result. It lacks two capabilities.
First, greedy cannot sacrifice power consumption for affinity gains. Its rule is
$\arg\min \Delta P_{i,m}$ at every step. It has no tolerance
gate. When a $G = 1$ candidate exists at
$\Delta P = (1 + \delta)\,\Delta P_{\min}^t$ for any
$\delta > 0$, greedy ignores it and selects the $G = 0$
candidate at $\Delta P_{\min}^t$.  The affinity opportunity
is lost permanently.  TGA captures it by accepting the bounded
cost $\delta \cdot \Delta P_{\min}^t$.

Second, greedy cannot recover power consumption after placement. Every
assignment $\pi(i) = m$ is final.  If a later placement makes
machine $m$ expensive (due to the cubic model), greedy cannot
move $i$ elsewhere.  The cost is locked in.  MIR can
retroactively correct this by moving $i$ to a cheaper machine,
recovering $|\Delta P(\rho_t)|$ watts.  Greedy has no
rearrangement phase, so its regret accumulates irreversibly.
Together, these two limitations create an asymmetry. So, it outperforms greedy approaches.

\section{Experimental Evaluation}
\label{sec:experiments}

We evaluate the proposed framework along three dimensions: placement quality relative to existing baselines, sensitivity to debate parameters, and token-cost efficiency.\footnote{The code is available at: \href{https://anonymous.4open.science/r/AgenticAI-in-Constrained-Multi-Objective-Container-Placement-E685/README.md}{https://anonymous.4open.science/r/AgenticAI-in-Constrained-Multi-Objective-Container-Placement-E685/README.md}} For the main experiments, we use Llama 3.1-8B \cite{grattafiori2024llama} as the primary LLM. All ablation studies and hyperparameter tuning are conducted with Llama 3.1 on a local GPU server\footnote{Intel(R) Xeon(R) Gold 5318Y CPU @ 2.10, GHz and an NVIDIA A100 80, GB PCIe GPU} using Ollama \cite{ollama2024}. In addition, we evaluate the same architecture with DeepSeek-R1 \cite{guo2025deepseek} using Ollama. We also report results with Claude Opus 4.6 and Claude Sonnet 4.6 \cite{anthropic2025claude} in extended-thinking mode through API-based inference.

%GLM-4.7-Flash \cite{glm2024chatglm}
%Qwen3-Next \cite{yang2025qwen3}
%and GPT-OSS
\subsection{Dataset and Experimental Setup}
We use the Google Cluster workload trace \cite{clusterdata:Wilkes2011} for heterogeneous machine capacities and application resource demands. Since the trace does not include placement preferences, we inject them synthetically using two binary matrices over application-machine pairs: $F$ (hard anti-affinity) and $G$ (soft affinity). We set $30\%$ of pairs as forbidden ($F=1$) and $20\%$ as preferred ($G=1$), sampled randomly with a fixed seed. 

Unless stated otherwise, we use $N{=}3$ debate rounds and allow each debating agent to propose $K{=}3$ candidate list to the Placement Agent. The Placement Agent tolerance is set to $\varepsilon{=}0.10$, meaning that the Placement Agent treats candidate placements with power consumption within 10\% of the minimum candidate power consumption as comparable. The number of machines and applications is fixed at 25 and 100, respectively. To minimize token consumption of LLMs, we set the maximum size of scratchpad to 10, candidate table to 30. All results are reported as improvement or degradation with respect to \textit{P-Greedy}, and the best-performing approach in each case is highlighted in bold in the tables.

%This configuration, evaluated using Llama~3.1, is treated as the reference setting for comparison with all other configurations. For power consumption, a change of $+x\%$ indicates an increase relative to the reference setting and is therefore undesirable, whereas $-x\%$ indicates a reduction and is therefore desirable. In contrast, for affinity satisfaction, $+x\%$ denotes an improvement over the reference setting and is desirable, whereas $-x\%$ denotes a degradation and is undesirable.

The temperature of LLM controls the randomness of token sampling during generation. In general, the temperature parameter ranges from $0$ to $2.0$, depending on the underlying model and inference API. Lower values lead to more deterministic outputs, whereas higher values increase diversity in the generated responses. 
We set the temperature to 0.3, which provides a balance between stable, reproducible outputs and limited exploration in the agents' reasoning process. All reported results are averaged over 3 independent runs. The variance across runs is less than 0.05\% on power consumption and less than 0.10\% for affinity satisfaction, confirming that the framework produces highly stable and reproducible placement decisions under this temperature setting.

\subsection{Comparison against Existing Heuristic Approaches}

We compare the proposed LLM-based agentic collaboration with general heuristics, specialized heuristics, and meta-heuristic approaches. As shown in \autoref{tab:methods}, the proposed framework consistently achieves the best overall balance between power consumption and affinity payoff. General heuristics such as \emph{First Fit}, \emph{Best Fit}, and \emph{P-Greedy} are limited because they rely on simple fixed rules that primarily emphasize feasibility or immediate local gain, without explicitly reasoning about the trade-off between energy efficiency and affinity satisfaction. As a result, their placements are less effective for the underlying multi-objective optimization problem.

Compared with the specialized heuristic \textit{HOP-CAPM}, the proposed method still performs better, indicating that manually designed heuristic rules remain limited when system conditions vary across lightly and heavily loaded machines. The meta-heuristic \textit{CPAAP} narrows the gap by refining placements through a broader search process, but its improvement is still constrained by the quality of the underlying heuristic search space. In contrast, the proposed framework benefits from explicit multi-agent reasoning, where different agents focus on complementary objectives and coordinate to produce better joint placement decisions. Overall, these results show that LLM-based agentic collaboration is more effective than traditional heuristic and meta-heuristic approaches for multi-objective placement.
The proposed approach outperforms \textit{P-Greedy}, \textit{HOP-CAPM} and \textit{CPAAP} by consuming 8.36\%, 7.12\% and 4.26\% power and acheives 39.30\%, 10.85\% and 1.83\% respectively.
\begin{table}[t]
\centering
\caption{\footnotesize \textbf{Comparison with heuristic and LLM-based methods. PC denotes total cluster power consumption (lower is better), and Aff denotes total affinity payoff (higher is better). (\textit{P-Greedy} as base)}}
\label{tab:methods}
\footnotesize
\renewcommand{\arraystretch}{1.03}
\setlength{\tabcolsep}{3pt}
\resizebox{\columnwidth}{!}{%
\begin{tabular}{@{}llcc@{}}
\toprule
\textbf{Category} & \textbf{Method} & \textbf{PC$\downarrow$} & \textbf{Aff$\uparrow$} \\
\midrule
\multirow{3}{*}{General Heuristics}
  & First Fit  & +16.69\% & -12.52\% \\
  & Best Fit   & +5.59\%  & +7.66\% \\
  & P-Greedy   & 0.00\%   & 0.00\% \\
\midrule
Specific Heuristics
  & HOP-CAPM   & -1.34\%  & +25.67\% \\
\midrule
Meta-Heuristics
  & CPAAP      & -4.28\%  & +36.79\% \\
\midrule
Proposed Agentic Collaboration
  & $N=3,\;K=5$ & \textbf{-8.36\%} & \textbf{+39.30\%} \\
\bottomrule
\end{tabular}%
}
\end{table}

\subsection{Example for Illustration}
\label{sec:example_mir}

We illustrate the placement process on a small example with four applications and three machines, using $N{=}1$, $K{=}1$, and $\varepsilon{=}0.10$. \autoref{tab:ex2_all}(A), (B), and (C) show the application demands, the initial machine state with capacities normalized to 1.0 for simplicity, and the affinity matrix $G(i,m)$, respectively.

We compare the power-greedy heuristic (\textit{P-greedy}) with the proposed framework. \textit{P-greedy} selects
$\arg\min_{(i,m)} \Delta P_{i,m}$ at each step, breaking ties using affinity, and performs no rearrangement. The proposed framework uses the Power Consumption Agent and Affinity Agent for debate,
the Placement Agent for TGA-based arbitration (Eq.~\ref{eq:tga}), and the Rearrangement Agent for MIR-based updates (Eq.~\ref{eq:mir}) after each placement. \autoref{tab:ex2_cand} presents the candidate tables across steps.

\begin{table}[t]
\centering
\caption{\footnotesize \textbf{(A) application demands, (B) machine environment (M2 is lightly loaded), (C) affinity matrix, and (D) final placement comparison between \textit{P-greedy} and the proposed approach (affinity satisfied placements are \underline{underlined}).}}
\label{tab:ex2_all}
\footnotesize
\renewcommand{\arraystretch}{1.08}

%==================== Row 1 ====================%
\begin{minipage}[t]{0.42\columnwidth}
\centering
\setlength{\tabcolsep}{4pt}
\begin{tabular}{@{}c|cccc@{}}
\toprule
$i$ & cpu & io & nw & mem \\
\hline
0 & 0.20 & 0.10 & 0.05 & 0.10 \\
1 & 0.25 & 0.10 & 0.05 & 0.10 \\
2 & 0.15 & 0.10 & 0.05 & 0.10 \\
3 & 0.20 & 0.10 & 0.05 & 0.10 \\
\bottomrule
\end{tabular}

\vspace{0.25em}
(A)
\end{minipage}
\hfill
\begin{minipage}[t]{0.54\columnwidth}
\centering
\setlength{\tabcolsep}{3pt}
\begin{tabular}{@{}c|rr|cccc@{}}
\toprule
$m$ & $P^{\text{idle}}$ & $P^{\text{max}}$
& $u^{\text{cpu}}$ & $u^{\text{io}}$
& $u^{\text{nw}}$ & $u^{\text{mem}}$ \\
\hline
0 & 45 & 200 & 0.00 & 0.00 & 0.00 & 0.00 \\
1 & 50 & 200 & 0.00 & 0.00 & 0.00 & 0.00 \\
2 & 55 & 200 & 0.10 & 0.05 & 0.00 & 0.05 \\
\bottomrule
\end{tabular}

\vspace{0.25em}
(B)
\end{minipage}

\vspace{0.7em}

%==================== Row 2 ====================%
\begin{minipage}[t]{0.34\columnwidth}
\centering
\setlength{\tabcolsep}{4pt}
\begin{tabular}{c c|ccc}
\toprule
& & \multicolumn{3}{c}{Machine ID ($m$)} \\
\cline{3-5}
& & m0 & m1 & m2 \\
\cline{1-5}
\multirow{4}{*}{\rotatebox{90}{App ID ($i$)}}
& 0 & 0 & 0 & 1 \\
& 1 & 0 & 0 & 1 \\
& 2 & 0 & 1 & 0 \\
& 3 & 0 & 1 & 0 \\
\bottomrule
\end{tabular}

\vspace{0.25em}
(C)
\end{minipage}
\hfill
\begin{minipage}[t]{0.62\columnwidth}
\centering
\setlength{\tabcolsep}{3pt}
\begin{tabular}{@{}c|cr|cr@{}}
\toprule
& \multicolumn{2}{c|}{P-Greedy} & \multicolumn{2}{c}{Proposed} \\
\cline{2-5}
$m$ & Placed app & $P$ (W) & Placed app & $P$ (W)\\
\hline
0 & 0 & 50.13   & 0,2  & 58.17 \\
1 & 1, \underline{2} & 64.97& \underline{3} & 54.96 \\
2 & 3 & 63.54   & \underline{1} & 65.15\\
\hline
Total & & 178.641 & & 178.28\\
\bottomrule
\end{tabular}
\vspace{0.25em}
(D)
\end{minipage}
\end{table}

\begin{table}[t]
\centering
\caption{\footnotesize \textbf{Candidate tables for the example. The top panel shows the candidate tables at $t{=}0$, $1$, and $2$, where both policies take the same action. The bottom panel shows the candidate tables at $t{=}3$, where \textit{P-greedy} and the proposed approach have different candidate table.}}
\label{tab:ex2_cand}

% ---------- Top row ----------
\scriptsize
\renewcommand{\arraystretch}{1.02}
\setlength{\tabcolsep}{2pt}

\resizebox{\columnwidth}{!}{%
\begin{tabular}{@{}c|rrrr|rrrr|rrrr@{}}
\toprule
\multicolumn{1}{c|}{}
& \multicolumn{4}{c|}{$t = 0$}
& \multicolumn{4}{c|}{$t = 1$ (after a2$\to$m1)}
& \multicolumn{4}{c}{$t = 2$ (after a0$\to$m0)} \\
\cline{2-13}
rk
& app & $m$ & $\Delta P$ & $G$
& app & $m$ & $\Delta P$ & $G$
& app & $m$ & $\Delta P$ & $G$ \\
\midrule
1
& 2 & 1 & 4.48 & 1
& 0 & 0 & 5.13 & 0
& 3 & 2 & 6.63 & 0 \\
2
& 2 & 0 & 4.63 & 0
& 3 & 0 & 5.13 & 0
& 1 & 2 & 8.24 & 1 \\
3
& 3 & 1 & 4.97 & 1
& 0 & 2 & 6.63 & 1
& 3 & 1 & 8.27 & 1 \\
4
& 0 & 1 & 4.97 & 0
& 3 & 2 & 6.63 & 0
& 3 & 0 & 10.34 & 0 \\
5
& 0 & 0 & 5.13 & 0
& 3 & 1 & 8.27 & 1
& 1 & 1 & 10.49 & 0 \\
\bottomrule
\end{tabular}%
}

\smallskip

% ---------- Bottom row ----------
{\normalsize
\renewcommand{\arraystretch}{0.98}
\setlength{\tabcolsep}{3pt}
\begin{tabular}{@{}c|rrrr|rrrr@{}}
\toprule
\multicolumn{1}{c|}{}
& \multicolumn{4}{c|}{$t = 3$ for P-Greedy}
& \multicolumn{4}{c}{$t = 3$ for Proposed} \\
\cline{2-9}
rk
& app & $m$ & $\Delta P$ (W) & $G$
& app & $m$ & $\Delta P$ (W) & $G$ \\
\midrule
1
& 1 & 1 & 10.49 & 0
& 1 & 1 & 12.85 & 0 \\
2
& 1 & 0 & 13.28 & 0
& 1 & 0 & 13.28 & 0 \\
3
& 1 & 2 & 18.13 & 1
& 1 & 2 & 15.09 & 1 \\
\bottomrule
\end{tabular}
}
\end{table}

\subsubsection{Placement trace} \paragraph{Step $t = 0$}
The candidate table (\autoref{tab:ex2_cand}, $t{=}0$) ranks
app2$\to$m1 first.
\textit{P-greedy} selects this row as the minimum $\Delta P$ with affinity tiebreak.
In the proposed framework, the Power Consumption Agent recommends app2$\to$m1 because it minimises power consumption. The Affinity Agent also recommends app2$\to$m1 because it achieves $G = 1$ at the lowest available $\Delta P$. Since both advocates agree, the Placement Agent selects app2$\to$m1 without contention. With only one instance placed, the Rearrangement Agent finds no improving operation and outputs \texttt{NoChange}.
\texttt{Placement state:\;} \textit{P-greedy}: \{a2$\to$m1\}, 
proposed framework: \{a2$\to$m1\}.

\paragraph{Step $t = 1$}
The candidate table (\autoref{tab:ex2_cand}, $t{=}1$) shows
app0$\to$m0 at rank~1. \textit{P-greedy} selects it as $\min \Delta P$.
In the proposed framework, the Power Consumption Agent recommends app0$\to$m0. The Affinity
Agent recommends app0$\to$m2, arguing that the affinity payoff justifies the extra cost.
The Placement Agent reads both arguments, computes the TGA gate ($4.97 \times 1.10 = 5.46$), and observes that the Affinity Agent's candidate at 6.86\,W exceeds the threshold.
It places app0$\to$m0, siding with the Power Consumption Agent. The Rearrangement Agent finds no eligible operation.
\texttt{Placement state:\;} \textit{P-greedy}: \{a2$\to$m1, a0$\to$m0\}, proposed framework: \{a2$\to$m1, a0$\to$m0\}.

\paragraph{Step $t = 2$}
The candidate table (\autoref{tab:ex2_cand}, $t{=}2$) shows app3$\to$m2 at rank~1. P-greedy selects it.
In the proposed framework, the Power Consumption Agent recommends app3$\to$m2. The Affinity Agent recommends app3$\to$m1, noting that machine~1 is the preferred target for app3. The Placement Agent computes the gate ($6.86 \times 1.10 = 7.54$) and finds 8.27 exceeds it. It places app3$\to$m2, identical to power-greedy. The Rearrangement Agent evaluates pairwise swaps among the three placed instances and identifies one eligible operation: swap app2\,(m1) $\leftrightarrow$ app3\,(m2). App2  moves to the more loaded m2, while app3 moves to the less loaded m1. This redistribution reduces cluster power by $0.71$\,W ($\Delta A = 0$).  
\texttt{Placement state:\;} \textit{P-greedy}: \{a2$\to$m1,
a0$\to$m0, a3$\to$m2\}, proposed framework: \{a2$\to$m2,
a0$\to$m0, a3$\to$m1\}.
From this point, the two policies face different machine utilizations, hence different
candidate tables at $t = 3$ (\autoref{tab:ex2_cand} ($t{=}3$, left and right).

\paragraph{Step $t = 3$}
Power-greedy sees the candidate table and chooses placing app1$\to$m1.
The Power Consumption Agent recommends app1$\to$m0. The Affinity Agent recommends app1$\to$m2. The Placement Agent computes the gate ($12.85 \times 1.10 = 14.14$) and finds 15.61 exceeds it. It places app1$\to$m0.

The Rearrangement Agent is activated again. Among four placed instances across three machines, it identifies the best operation: swap app2\,(m2) $\leftrightarrow$
app1\,(m0). App1 leaves the now-packed m0 and moves to m2. App2 moves from m2 to m0. Replacing the heavy app with the light one on the loaded machine sharply reduces the cubic penalty. This swap saves $2.21$\,W in cluster power and yields $\Delta A = +1$ because $G(1,2) = 1$.
\texttt{Final placement:\;} P-greedy: \{a2$\to$m1,
a0$\to$m0, a3$\to$m2, a1$\to$m1\}, proposed framework:
\{a0$\to$m0, a2$\to$m0, a3$\to$m1, a1$\to$m2\}.

\subsubsection{Final comparison}

\autoref{tab:ex2_all} (D) summarises the outcomes.  Power-greedy
achieves a cluster power of $50.13 + 64.97 + 63.54 =
178.64$\,W with 1 affinity satisfied placement. The proposed framework achieves $58.17 + 54.96 + 65.15 = 178.28$\,W with 2 affinity satisfied placement. The proposed framework outperforms the dual optimization of P-greedy: 0.36\,W lower power and
one additional affinity point.

In larger instances with denser affinity matrices, the Placement Agent will encounter cases where $G = 1$ candidates fall within the tolerance gate. In those cases, Placement Agent will select them directly, contributing additional affinity that
power-greedy cannot.

\subsection{Comparison of Different LLM-based Architecture}

As shown in \autoref{tab:arch_ablation}, the proposed architecture provides the best overall balance between power consumption and affinity payoff, and is therefore used as the reference configuration. Simpler architectures reduce token usage and the number of LLM calls, but this comes at the expense of weaker placement quality. In particular, collapsing the framework into a single agent or a small number of agents limits the system’s ability to reason explicitly about competing objectives, which leads to less effective trade-offs between energy efficiency and affinity satisfaction.

Introducing specialized agents improves the quality of decisions by separating the optimization objectives and enabling more structured coordination. In particular, the inclusion of dedicated Power and Affinity agents strengthens multi-objective reasoning, while the Placement Agent arbitrates between their recommendations under the tolerance-based policy. The confidence-scoring variant~\cite{xiong2024can,tian2023just}, where agents provide recommendations together with confidence scores instead of explicit debate arguments, improves coordination compared with simpler baselines but still fails to match the full framework. This suggests that confidence values alone are not sufficient to capture the richer interactions and trade-offs exposed through argument-based debate.

Overall, the results indicate that explicit role specialization, debate-based coordination, and post-placement refinement are all important for achieving strong placement quality. Although the full architecture incurs higher token cost and more LLM calls per step, this additional reasoning overhead enables consistently better joint optimization of power consumption and affinity payoff.

As shown in \autoref{fig:comparison}, three configurations (\textit{P-Greedy}, Three Agents and proposed architecture) follow nearly identical power
curves. At low utilization, the cubic term
contributes little to power. Placement
decisions have small $\Delta P$ values with narrow spread, so debate and rearrangement offer minimal advantage. But differences in affinity satisfaction emerge: the proposed architecture accumulates affinity earlier because TGA actively searches the $\varepsilon$-band for $G{=}1$ candidates even when the power savings are small.
As machines fill, the cubic
marginal cost grows quickly. Three Agents and \textit{P-Greedy} lack rearrangement, so suboptimal early placements permanently occupy low-cost slots.  Future steps face uniformly high
$\Delta P$ values. This highlights the importance of Rearrangement Agent.

\begin{table}[t]
\centering
\caption{\footnotesize \textbf{Comparison across different LLM-agentic architectures (\textit{P-Greedy} as base). \protect\footnotemark}}
\label{tab:arch_ablation}
\scriptsize
\renewcommand{\arraystretch}{1.03}
\setlength{\tabcolsep}{3pt}
\resizebox{\columnwidth}{!}{%
\begin{tabular}{@{}lrrrc@{}}
\toprule
\textbf{Agent Architecture} &
\textbf{PC$\downarrow$} &
\textbf{Aff$\uparrow$} &
\textbf{Tokens} &
\textbf{Calls/step} \\
\midrule
Single Agent                & +4.12\%          & +14.22\%          &   201,366 & 1 \\
Two Agents (PlaA+RA)        & -0.97\%          & +19.85\%          &   340,771 & 2 \\
Three Agents (AA+PowA+PlaA) & +2.89\%          & +20.07\%          & 1,860,101 & 3 \\
Confidence Scoring          & -3.17\%          & +36.56\%          &   712,460 & 4 \\
Proposed Architecture       & \textbf{-8.36\%} & \textbf{+39.30\%} & 1,937,743 & 8 \\

\bottomrule
\end{tabular}%
}
\end{table}
\footnotetext{\label{note1}PC denotes power consumption (lower is better), Aff denotes affinity payoff (higher is better), Tokens denote the total tokens consumed over the full scheduling run, and Calls/step denotes the number of LLM invocations per placement step. PowA, AA, PlaA, and RA denote the Power Consumption Agent, Affinity Agent, Placement Agent, and Rearrangement Agent, respectively.}
\setcounter{sharedtabfn}{\value{footnote}}

\begin{figure}
    \centering
    \includegraphics[width=0.95\linewidth]{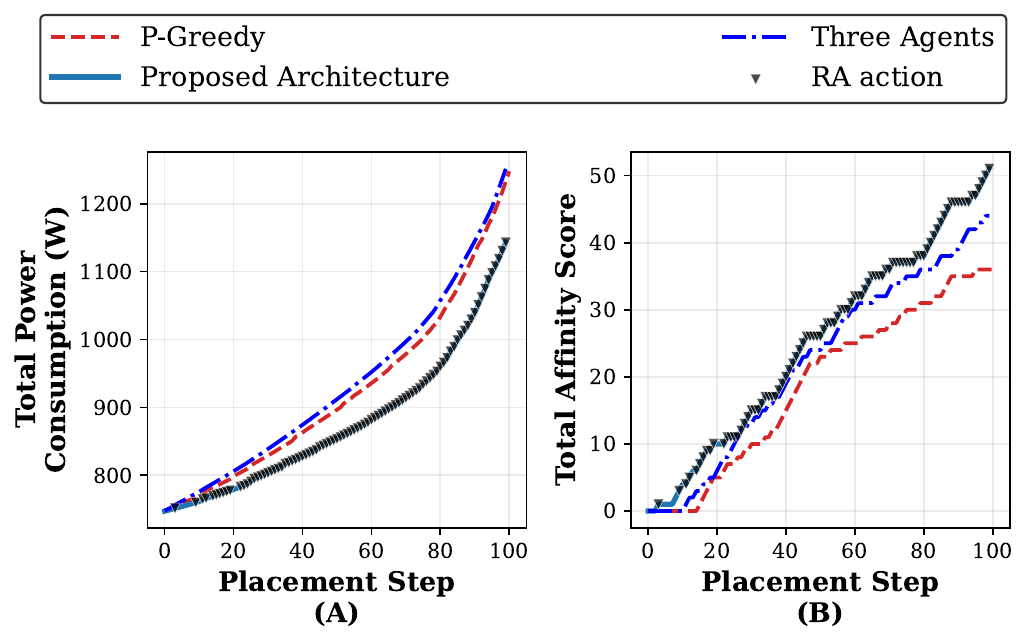}
    \caption{\footnotesize \textbf{(A) Cumulative total power consumption, and (B) cumulative affinity score over 100 placement steps.
    Triangles mark steps where the RA executed a
    Pareto-improving operation.}}
  \label{fig:comparison}
\end{figure}
\subsection{Effect of Variation of Tolerance}

\begin{figure}[tb!]
    \centering
    \includegraphics[width=0.95\linewidth]{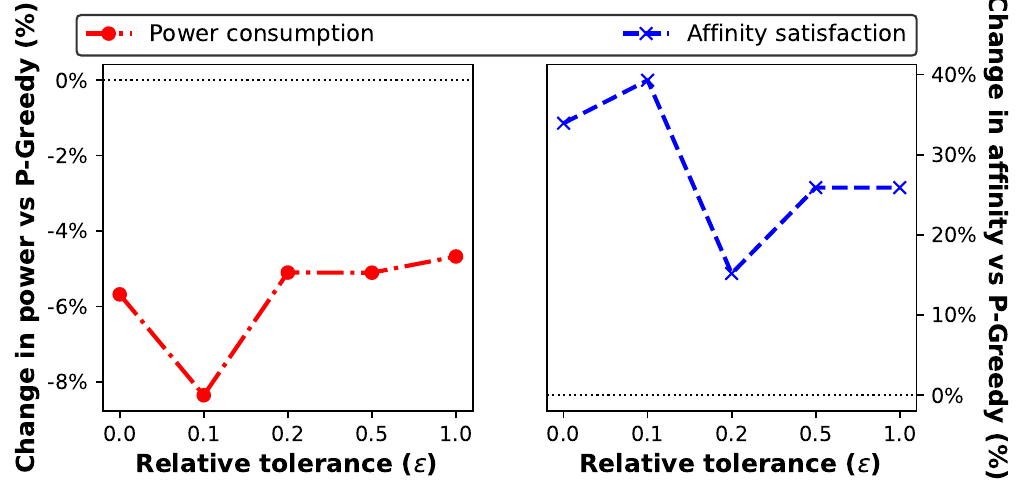}
    %\vspace{-0.5cm}
    \caption{\footnotesize \textbf{Comparison of power consumption and affinity satisfied placement with variation of tolerance ($\varepsilon$) of Placement Agent.}}
    \vspace{-0.5cm}
    \label{fig:epsilon}
\end{figure}

In \autoref{fig:epsilon}, the baseline performance is achieved at $\varepsilon{=}0.1$, which provides the best balance between power consumption and affinity satisfaction. When $\varepsilon$ is set too low or too high relative to this baseline, performance degrades. In particular, $\varepsilon{=}0$ makes the selection process overly rigid by restricting the Placement Agent to only the top-ranked candidates, thereby reducing flexibility and making the framework behave similarly to a fixed-rule heuristic. In contrast, larger tolerance values admit a broader set of inferior candidates, which weakens the quality of agent recommendations and makes the final arbitration less focused. A small non-zero tolerance, therefore, is most effective because it preserves decision quality while still allowing enough flexibility for better long-term placements and subsequent MIR-based rearrangements.

\begin{table}[t]
\centering
\caption{\footnotesize \textbf{Ablation over the number of debate rounds ($N$) and candidate list size ($K$) (\textit{P-Greedy} as base).\textsuperscript{\ref{note1}}}}
\label{tab:ablation}
\begin{tabular}{@{}ccrrrc@{}}
\toprule
$N$ & $K$ & \textbf{PC\,$\downarrow$}
  & \textbf{Aff\,$\uparrow$}
  & \textbf{Tokens}
  & \textbf{Calls/step} \\
\midrule
0 & 5   & -2.11\% & +56.86\% &    325,063 &  4 \\
1 & 5   & -4.25\% & \textbf{+58.52\%} &    752,249 &  4 \\
2 & 5   & -5.86\% & +52.75\% & 1,271,956 &  6 \\
3 & 5   & \textbf{-8.36\%} & +39.30\% & 1,937,743 & 8 \\
\midrule
3 & 1   & -4.00\% & +47.28\% &  1,590,731 &  8 \\
3 & 5   & \textbf{-8.36\%} & +39.30\% &    1,937,743 & 8 \\
3 & all & \textbf{-8.36\%} & +39.30\% &    3,329,739 &  8 \\
\bottomrule
\end{tabular}
\end{table}

\subsection{Effect of Debate Rounds and Candidate List Size}

As shown in \autoref{tab:ablation}, the reference setting $(N,K)=(3,5)$ achieves the best overall trade-off between power consumption and affinity payoff. Reducing the number of debate rounds lowers the reasoning cost (tokens), but also weakens decision quality, since the agents have less opportunity to refine the trade-off between competing objectives. As a result, fewer debate rounds tend to produce more short-sighted placements, while additional rounds improve coordination and lead to better joint decisions.

Similarly, restricting the candidate list size reduces flexibility by exposing the Placement Agent to only a small set of alternatives, which limits the quality of the final choice. Increasing the list size improves decision quality up to a moderate level, after which the gains saturate. Beyond that point, larger candidate lists mainly increase token usage without improving placement quality. This shows that a moderate number of debate rounds and a moderate candidate list size are sufficient to balance reasoning cost and performance effectively.

\subsection{Scalability Analysis} 

As shown in \autoref{tab:scalability}, the proposed framework maintains the best overall performance across all tested problem sizes, indicating that its advantage is not limited to a specific scale of applications or machines. The same general trend observed earlier continues to hold as the system grows: rule-based heuristics remain limited by fixed local decision rules, while refinement-based methods improve the results only partially. In contrast, the proposed agentic framework continues to produce stronger joint placement decisions because it explicitly balances power consumption and affinity satisfaction through coordinated multi-agent reasoning. An important observation is that the relative behavior of the compared methods remains consistent as the workload and infrastructure size increase, which suggests that the proposed design scales robustly with problem size. Although \textit{CPAAP} achieves slightly better affinity in a few cases, this advantage is marginal and should be interpreted in light of the fact that affinity is treated as a soft objective, whereas reducing power consumption remains the primary optimization target. Overall, the results show that the proposed method preserves its effectiveness across varying numbers of applications and machines while maintaining the best overall trade-off between the two objectives.

\begin{table}[t]
\centering
\caption{\footnotesize \textbf{Scalability comparison across different problem sizes. Each tuple $(A,M)$ denotes the number of applications $A$ and machines $M$, respectively. PC denotes total cluster power consumption (lower is better), and Aff denotes total affinity payoff (higher is better) (\textit{P-Greedy} as base).}}
\label{tab:scalability}
\scriptsize
\setlength{\tabcolsep}{2.5pt}
\renewcommand{\arraystretch}{1.03}
\resizebox{\columnwidth}{!}{%
\begin{tabular}{@{}l|ccc|ccc@{}}
\toprule
\multirow{2}{*}{\textbf{Method}} &
\multicolumn{3}{c|}{\textbf{PC$\downarrow$}} &
\multicolumn{3}{c}{\textbf{Aff$\uparrow$}} \\
\cline{2-4}\cline{5-7}
& $(200,50)$ & $(400,150)$ & $(500,200)$
& $(200,50)$ & $(400,150)$ & $(500,200)$ \\
\hline
\hline
First Fit & +14.99\% & +6.23\%  & +15.28\% & +35.20\% & -1.35\%   & +53.49\%  \\
Best Fit  & +4.40\%  & +1.82\%  & +2.38\%  & +0.18\%  & +7.74\%   & +37.66\%  \\
P-Greedy  & 0.00\%   & 0.00\%   & 0.00\%   & 0.00\%   & 0.00\%    & 0.00\%    \\
HOP-CAPM  & -4.22\%  & -6.53\%  & -4.58\%  & +62.14\% & +41.81\%  & +155.05\% \\
CPAAP     & -6.63\%  & -9.39\%  & -5.49\%  & \textbf{+78.46\%} & \textbf{+53.06\%}  & +165.59\% \\
Proposed  & \textbf{-10.94\%} & \textbf{-13.72\%} & \textbf{-10.08\%}
          & +75.13\% & +51.72\% & \textbf{+168.82\%} \\
          \bottomrule
\end{tabular}%
}
\vspace{-0.4cm}
\end{table}

\subsection{Comparison across Different LLMs}

As shown in \autoref{tab:llm_compare}, stronger LLMs consistently yield better placement quality under the same agentic setting. While the overall token usage remains comparable across models, the quality of reasoning has a clear impact on both power consumption and affinity payoff. In particular, larger and more capable models produce better joint decisions because they can more effectively interpret the candidate set, reason about competing objectives, and select placements with better long-term benefit. Among the evaluated models, Claude-based variants achieve the strongest overall results, with Claude Opus performing best, while DeepSeek-R1 also improves substantially over the Llama~3.1 baseline. These results indicate that the effectiveness of the proposed framework depends not only on the agentic design, but also on the reasoning capability of the underlying LLM.

\begin{table}[t]
\centering
\caption{\footnotesize \textbf{Comparison of different LLMs under the same agentic setting. All results are computed with $N=1$ and $K=1$ due to API cost constraints for proprietary models. This reduced configuration is applied uniformly across all LLMs in this table to ensure a fair cross-model comparison. Note that these results are not directly comparable to \autoref{tab:methods}, which uses the full $N=3$ and $K=5$ reference configuration (\textit{P-Greedy} as base).\textsuperscript{\ref{note1}}}}

\label{tab:llm_compare}
\footnotesize
\setlength{\tabcolsep}{3pt}
\renewcommand{\arraystretch}{1.03}
\resizebox{\columnwidth}{!}{%
\begin{tabular}{@{}lrrrr@{}}
\toprule
\textbf{LLM} & \textbf{Parameters}  & \textbf{PC$\downarrow$} & \textbf{Aff$\uparrow$} & \textbf{Tokens}\\
\midrule
 Llama 3.1         & 8B    & -3.36\%              & +35.30\%            & 854,547\\
 DeepSeek-R1       & 7B    & -14.48\%             & +69.13\%            & 934,472\\
 Claude Opus 4.6   & 100B+ & \textbf{-16.63\%}    & \textbf{+74.75\%}   & 890,555\\
 Claude Sonnet 4.6 & 10B+  & -13.77\%             & +59.19\%            & 911,190\\
\bottomrule
\end{tabular}%
}
\vspace{-0.4cm}
\end{table}

\subsection{Token Cost and Practical Overhead}

The reference configuration $(N,K)=(3,5)$ consumes approximately 2M tokens for a full scheduling run with 100 placement decisions. Using Llama~3.1-8B served locally through Ollama on a single GPU, each LLM invocation takes about 7--8 seconds on average, resulting in an end-to-end scheduling time of roughly 90 minutes.

Despite this overhead, the inference cost remains small relative to the operational savings obtained from improved placements. In our setting with 100 applications and 25 machines, the proposed framework reduces power consumption by about 4.26\% over best existing heuristics. The improvement is more when a better LLM is used. Using current cloud pricing as a rough reference, the token cost of one complete
run is approximately \$0.5 for Llama~3.1-8B and about \$10 for Claude~\cite{artificialanalysis2025,anthropic2025pricing}. By contrast, the corresponding reduction in power consumption can translate to around \$100 of operational savings over a representative
execution horizon, depending on the electricity price and workload duration.

This gap becomes even more significant at larger scales. A datacenter typically hosts thousands of servers, depending on rack density and deployment configuration~\cite{crs2025datacenter,rpa2025datacenter}, with each server running multiple containerized applications over long durations. Since annual energy costs in datacenters can range from around one million dollar for smaller facilities to tens of millions of dollars for large-scale deployments \cite{ahmed2021review,eia2025electric,lbnl2024datacenter}. So, a power reduction of 4.26\% translates into substantial operational savings. By comparison, the LLM-based scheduling cost remains relatively small, typically on the order of a few hundred to low thousands of dollars.

\section{Conclusion}

We proposed a four-agent collaboration framework that outperforms existing heuristic and meta-heuristic baselines on both power consumption and
affinity payoff. The natural next step is reducing the inference cost and latency of the agentic pipeline. To this end, future work will pursue three directions. First, we will investigate distilling the learned debate patterns into lightweight policy models that replicate multi-agent
reasoning at a fraction of the inference cost.  Second, we will explore prompt compression and response caching strategies to reduce per-step token consumption while preserving decision quality. Third, we aim to extend the
framework to dynamic workloads where applications arrive and depart over time, requiring online rescheduling.

%\footnotesize
%\setstretch{0.89}
\bibliographystyle{IEEEtran}
\bibliography{main}
\end{document}